\documentclass[12pt]{amsart}

\usepackage{graphicx, verbatim,amsmath,amssymb}

\newcommand{\nd}{{\noindent}}

\newcommand{\cut}{\hbox{\hskip 1pt\vrule width3pt height2pt depth1pt \hskip1pt}}

\newcommand{\www}{{\mathcal W}}

\newtheorem{thm}{Theorem}[section]
\newtheorem*{thm*}{Theorem}
\newtheorem{cor}[thm]{Corollary}

\everymath{\displaystyle}

\theoremstyle{definition}

\begin{document} 

\title{Phase transitions in a complex network}
\author{Charles Radin and Lorenzo Sadun}

\address{Charles Radin\\Department of Mathematics\\The University of
  Texas at Austin\\ Austin, TX 78712} \email{radin@math.utexas.edu}
\address{Lorenzo Sadun\\Department of Mathematics\\The University of
  Texas at Austin\\ Austin, TX 78712} \email{sadun@math.utexas.edu}
\thanks{This work was partially supported by NSF
  grants DMS-1208941 and DMS-1101326} 

\begin{abstract}
We study a mean field model of a complex network, focusing on
edge and triangle densities. Our first result is the
derivation of a variational characterization of the entropy density,
compatible with the infinite node limit. We then determine the
optimizing graphs for small triangle density and a range of edge
density, though we can only prove they are local, not global, maxima
of the entropy density. With this assumption we then prove that the
resulting entropy density must lose its analyticity in various
regimes.  In particular this implies the existence of a phase
transition between distinct heterogeneous multipartite phases at low
triangle density, and a phase transition between these phases and the
disordered phase at high triangle density.
\end{abstract}

\maketitle

\section{Introduction}

Exponential random graph models are a well known class of complex
networks; see \cite{Ne} and references therein. Using the language of
statistical mechanics they are mean field models, in the grand
canonical ensemble, with a variety of possible many-body interactions
appropriate to the model's use. Phase transitions, which require an
infinite node limit, have been proven for them \cite{CD, RY} using the
recently developed `graphon' formalism \cite{Lov} in place of the
infinite volume limit formalism \cite{R1, R2} used in statistical
mechanics. Exponential random graph models are mean field models and
therefore the analogues of the various statistical mechanics ensembles
(microcanonical, grand canonical, pressure, $\ldots$) which are
equivalent in the infinite volume limit for particle systems with
short range interactions \cite{R1}, need not be equivalent in these
mean field models; see for instance \cite{TET}. (Equivalence of
ensembles is discussed further in the Conclusion.) In this work we use
the microcanonical ensemble of one of the best known exponential
random graph models, one originally formulated by Strauss \cite{St},
and give evidence of phase transitions which are not as accessible in
the grand canonical ensemble. The transitions previously analyzed for
a wide class of exponential random graphs are similar to liquid/gas
transitions in that they are transitions between graphs of similar
character, of the same (fluid-like) phase \cite{RY}, while the
transitions we focus on in the microcanonical ensemble are analogous
to solid/solid transitions, transitions between graphs of different
phases. (See \cite{AR} for a more primitive grand canonical analysis
of these phases.)

We need some network notation. Consider the set $\hat G^n$ of simple
graphs $G$ with set $V(G)$ of (labelled) vertices, edge set $E(G)$ and
triangle set $T(G)$, where the cardinality $|V(G)|=n$. (`Simple' means
the edges are undirected and there are no multiple edges or loops.)
Think of an unordered pair of vertices as a point in an abstract
space, an edge as a particle that may occupy that point, and a
triangle as a many-body interaction energy associated with its edges,
so the microcanonical partition function, $\displaystyle
Z^{n,\delta}_{e,t}$, is the number of simple graphs such that:
\begin{equation} e(G)\equiv \frac{|E(G)|}{{n \choose 2}}
\in (e-\delta,e+\delta) \quad \hbox{ and } \quad
t(G)\equiv \frac{|T(G)|}{{n\choose 3}} \in (t-\delta,t+\delta).
\end{equation} 
Graphs in $\displaystyle \cup_{n\ge 1}\hat G^n$ are known to have edge
and triangle densities, $(e,t)$, dense in the region $R$ in the
$e,t$-plane bounded by three curves, $c_1: (e,e^{3/2}), \ \ 0\le e\le
1$, the line $l_1: \ (e,0), \ \ 0\le e\le 1/2$ and a certain scalloped
curve $(e,h(e)),\ \ 1/2\le e\le 1$, lying above the curve
$(e,e(2e-1), \ \ 1/2\le e\le 1$, and meeting it when
$e=e_k=k/(k+1),\ \ k\ge 1$; see \cite{PR} and references therein, and Figure 1.
\vskip.3truein
\begin{figure}[h]
\includegraphics[width=3in]{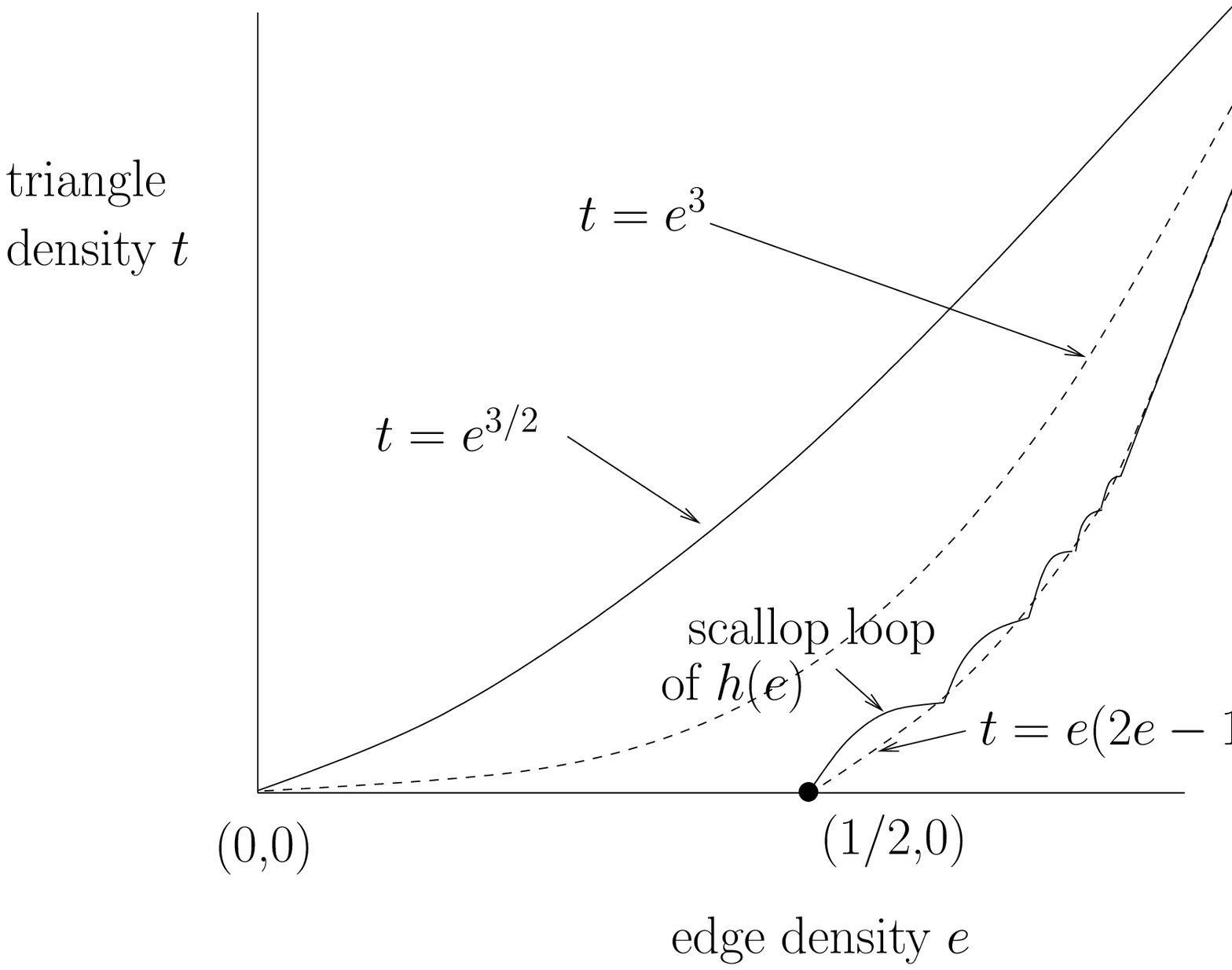}
\caption{The microcanonical phase space $R$, outlined 
in solid lines}
\label{phasefig}
\end{figure}

We are interested in the relative density of graphs in $R$, more
precisely in the entropy, the exponential rate of growth of the number
of graphs as $n$ grows, as follows. First consider
\begin{equation}
  s^{n,\delta}_{e,t}=\frac{\ln(Z^{n,\delta}_{e,t})}{n^2}, \hbox{ and
    then } 
s_{e,t}=\lim_{\delta\to 0^+}\lim_{n\to \infty}s^{n,\delta}_{e,t}.
\end{equation} 
\nd (The existence of the double limit will be proven.)
We will measure the growth rate by the entropy density
$s_{e,t},$ and the main question of interest for us is the existence
of phase transitions ({\it i.e.} lack of analyticity of $s_{e,t}$) near the
lower boundary of $R$ in Figure 1. The lower boundary consists of the
scalloped curve together with the `first scallop', the line from
$(0,0)$ to $(1/2,0)$.

We now need to review some notation and results concerning graphons,
as recently developed in \cite{LS1,
  LS2, BCLSV, BCL, LS3}. See also the recent book \cite{Lov}.

\section{Graphons}

Consider the set $\www$ of all symmetric, measurable functions 
\begin{equation} g:(x,y)\in [0,1]^2\to g(x,y)\in [0,1].\end{equation} 
Think of each axis as a continuous set of vertices of a graph. For a
graph $G\in \hat G^n$ one associates 
\begin{equation} \label{checkerboard} g^G(x,y) = \begin{cases} 1 &\hbox{if }(\lceil nx \rceil , \lceil ny \rceil
)\hbox{ is an edge of }G\cr 0 & \hbox{otherwise,} \end{cases}
\end{equation} 
where $\lceil y \rceil$ denotes the smallest integer greater than or
equal to $y$. 
For $g\in \www$ and simple graph $H$ we define
\begin{equation} t(H,g)\equiv \int_{[0,1]^\ell} \prod_{(i,j)\in E(H)}g(x_i,x_j)\,dx_1\cdots
   dx_\ell, \end{equation} 
where $\ell = |V(H)|$, and note that for a graph $G$, $t(H,g^G)$ is the 
density of graph homomorphisms $H\to G$:
\begin{equation} \frac{|\hbox {hom}(H,G)|}{|V(G)|^{|V(H)|}}. \end{equation} 
We define an equivalence relation on $\www$ as follows: $f\sim g$
if and only if $t(H,f)=t(H,g)$ for every simple graph $H$.  Elements
of $\www$ are called ``graphons'', elements of the quotient space $\tilde \www$ are called ``reduced graphons'', and
the class containing $g\in \www$ is denoted $\tilde g$. 
Equivalent functions in $\www$ differ by a change of variables in the
following sense. Let $\Sigma$ be the space of measure preserving maps
$\sigma: [0,1]\to [0,1]$, and for $f$ in $\www$ and $\sigma\in
\Sigma$, let $f_\sigma(x,y)\equiv f(\sigma(x),\sigma(y))$. Then $f\sim
g$ if and only if there exist $\sigma, \sigma'$ in $\Sigma$ such that
$f_\sigma =g_{\sigma'}$ almost everywhere; see Cor. 2.2 in \cite{BCL}. The space
$\www$  is compact with respect to the `cut metric'
defined as follows.
First, on ${\www}$ define:
$${d}_{\cut}(f,g)\equiv \sup_{S,T\subseteq [0,1]}\Big| \int_{S\times
  T}[f(x,y)-g(x,y)]\, dxdy\Big|. $$
Then on $\tilde \www$ define the cut metric by:

$${\tilde d}_{\cut}(\tilde f,\tilde g)\equiv \inf_{\sigma,\sigma'\in \Sigma}
{d}_{\cut}(f_{\sigma},g_{\sigma'}). $$

We will use the fact, which follows easily from Lemma 4.1 in 
\cite{LS1}, that the cut metric is equivalent to the metric
\begin{equation} \delta_{_{\hbox{hom}}}(\tilde f,\tilde g)\equiv \sum_{j\ge 1}
\frac{1}{2^j}|t(H_j,f)-t(H_j,g)|, 
\end{equation} 
where $\{H_j\}$ is a countable set of simple graphs, one from each
graph-equivalence class. 
Also note that
if each vertex of a finite graph is split into the same number of
`twins', each connected to the same vertices, the result stays in the
same equivalence class, so for a convergent sequence $\tilde g^{G_j}$ one may
assume $|V(G_j)|\to \infty$. 

The value of this graphon formalism here is that one can use large deviations
on graphs with independent edges, proven in \cite{CV}, to give an optimization
formula for $s_{e,t}$, which allows us to analyze $s_{e,t}$ near the
graphons of minimal triangle density, the lower boundary of $R$ in
Figure 1. We next use the large deviations rate function for graphs
with independent edges to give a variational
characterization for the entropy density. (There is a variational
characterization in \cite{CD} of the free energy density in the grand
canonical ensemble.)

\section{A variational characterization of the entropy density}

\begin{thm}\label{thm1} 
  For any possible pair $(e,t)$, $s_{e,t} = - \min I(g)$, where the
  minimum is over all graphons $g$ with $e(g)=e$ and $t(g)=t$,
  where 
\begin{equation} e(g)=\int_0^1 \!\!\! \int_0^1 g(x,y) \, dx \, dy, \qquad t(g) =
  \int_0^1\!\!\!\int_0^1\!\!\!\int_0^1 g(x,y) g(y,z) g(z,x) \, dx \, dy \, dz,$$
  and the rate function is $$I(g) = \int_0^1\!\!\!\int_0^1 I_0(g(x,y)) \,dx\, dy, \hbox{
    \rm where } I_0(u)= \frac{1}{2} \left [u \ln(u) +
    (1-u)\ln(1-u)\right ].
\end{equation} 
\end{thm}

\begin{proof} We first prove that $s_{e,t}$ is well-defined.
  A priori all we know is that $\liminf\ln
  (Z_{e,t}^{n,\delta})/n^2$ and $\limsup\ln (Z_{e,t}^{n,\delta})/n^2$
  exist as ${n\to\infty}$. However, we will show that they both
  approach $-\min I(g)$ as $\delta \to 0^+$.

  We need to define a few sets. Let $U_\delta$ be the set of graphons
  $g$ with $e(g)$ and $t(g)$ strictly within $\delta$ of $e$ and $t$, i.e. the
  preimage of an open square of side $2\delta$ in $(e,t)$-space, and
  let $F_\delta$ be the preimage of the closed square.  Let $\tilde U_\delta$ and $\tilde F_\delta$ be the
  corresponding sets in $\tilde \www$.  Let
  $|U_\delta^n|$ and $|F_\delta^n|$ denote the number of graphs with
  $n$ vertices whose checkerboard graphons (\ref{checkerboard}) lie in $U_\delta$ or
  $F_\delta$. The large deviation principle, Theorem 2.3 of
  \cite{CV}, implies that:
\begin{equation} \limsup_{n \to \infty} \frac{\ln|F_\delta^n|}{n^2} \le -\inf_{\tilde g\in \tilde F_\delta} I(\tilde g),\end{equation} 
which also equals $-\inf_{g \in F_\delta} I(g)$, and that
\begin{equation} \liminf_{n \to \infty} \frac{\ln|U_\delta^n|}{n^2} \ge -\inf_{\tilde g \in \tilde U_\delta} I(\tilde g),\end{equation} 
which also equals $-\inf_{g \in U_\delta} I(g)$.  This yields a chain
of inequalities
\begin{equation} -\inf_{U_\delta} I(g) \le \liminf \frac{\ln|U^n_\delta|}{n^2} \le \limsup \frac{\ln|U^n_\delta|}{n^2} \le 
\limsup \frac{\ln|F_\delta^n|}{n^2} \le - \inf_{F_\delta} I(g) \le -
\inf_{U_{\delta+\delta^2}} I(g)\end{equation} As $\delta \to 0^+$, the limits of
$-\inf_{U_\delta} I(g)$ and $-\inf_{U_{\delta+\delta^2}} I(g)$ are the
same, and everything in between is trapped. 

So far we have proven that 
\begin{equation} s_{e,t} = -\lim_{\delta \to 0^+}
\inf_{U_\delta} I(g).\end{equation} Next we must show that the right hand side is
equal to $- \min_{F_0} I(g)$.  By definition, we can find a sequence
of reduced graphons $\tilde g_\delta \in \tilde U_\delta$ such that
$\lim_{\delta \to 0} I(\tilde g_\delta) = \lim \inf_{U_\delta}
I(g)$. Since $\tilde W$ is compact, these reduced graphons converge to
a reduced graphon $\tilde g_0$, represented by a graphon $g_0 \in
F_0$. Since $I$ is lower-semicontinuous \cite{CV}, $I(g_0) \le \lim
I(g_\delta)$, so $\min_{F_0} I(g) \le \lim \inf_{U_\delta} I(g)$. (We
write $\min$ rather than $\inf$ since $\tilde F_0$ is compact.)
However, $\min_{F_0} I(g)$ is at least as big as $\inf_{U_\delta}
I(g)$, since $F_0 \subset U_\delta$.  Thus $\min_{F_0} I(g) =
\lim_{\delta \to 0} \inf_{U_\delta} I(g)$.
\end{proof} 

\section{Minimizing the rate function on the boundary}

>From now on we will work exclusively with graphons rather than with
graphs. From Theorem \ref{thm1}, all questions boil down
to ``minimize the rate function over such-and-such region''. The first
region we study is the lower boundary of $(e,t)$-space, beginning with
the first (flat) scallop:

\begin{thm}\label{flat} If $e \le 1/2$ and $t=0$, then $\min_{F_0}
  I(g) = I_0(2e)/2$, and this minimum is achieved at the graphon
\begin{equation} \label{2fold} g_0(x,y) = \begin{cases}2e &\hbox{if }x<\frac{1}{2} < y \hbox{  or } y<\frac{1}{2}<x; \cr 0 & \hbox{otherwise.}\end{cases}\end{equation} 
Furthermore, any other minimizer is equivalent to $g_0$, corresponding
to the same reduced graphon. \end{thm}

\begin{proof} Since $t(g)$ is identically zero, the measure of the set
  $\{(x,y) \in [0,1]^2 | g(x,y) =0\}$ is at least 1/2.  Otherwise,
  the graphon $\bar g(x,y) = \begin{cases}1 & \hbox{if } g(x,y)>0; \cr
    0& \hbox{otherwise,}\end{cases}$ would have no triangles and would
  have edge density greater than 1/2, which is impossible.  So we
  restrict attention to graphons that are zero on a set of measure at
  least 1/2 and have edge density $e$. From the convexity of $I_0$, 
  we know that the graphon minimizing
  $I$ must be zero on a set of measure 1/2 and
  must be constant on the rest. Thus $g_0$ is a minimizer, and $\min_{F_0}I(g) = I(g_0) = 
I_0(2e)/2$.

  Now suppose that $g$ is another minimizer. Since $g$ is zero on a
  set of measure 1/2 and is $2e$ on a set of measure 1/2, $\bar g$ is
  1 on a set of measure 1/2, and so describes a graphon with edge
  density 1/2 and no triangles. This means that $\bar g$ describes a
  complete bipartite graph with the two parts having the same
  measure. That is, $\bar g$ is equivalent to the graphon that equals
  1 if $x<\frac{1}{2}<y$ or $y< \frac{1}{2} < x$ and is zero
  everywhere else. But then $g=2e\bar g$ is equivalent to $g_0$.
\end{proof}

The situation on the curved scallops is slightly more complicated. Pick an
integer $\ell>1$. (The case $\ell=1$ just gives us our first scallop.) If $e
\in \left[1-\frac{1}{\ell}, 1-\frac{1}{\ell+1}\right ]$, then any graph $G$
with edge density $e$ and the minimum number of triangles has to take
the following form (see \cite{PR} for the history).  Let
\begin{equation}  c = \frac{\ell + \sqrt{\ell(\ell-e(\ell+1))}}{\ell(\ell+1)}.\end{equation} 
There is a partition of $\{1,\ldots,n\}$ into $\ell$ pieces, the first $\ell-1$
of size $\lfloor cn\rfloor$ and the last of size between $\lfloor cn\rfloor$ 
and $2\lfloor cn\rfloor$, such
that $G$ is the complete $\ell$-partite graph on these pieces, plus a
number of additional edges within the last piece. ($\lfloor y\rfloor$
denotes the largest integer greater than or equal to $y$.)
These additional
edges can take any form, as long as there are no triangles within the
last piece.

This means that, after possibly renumbering the vertices, the graphon
for such a graph can be written as an uneven $\ell \times \ell$ checkerboard
obtained from cutting the unit interval into pieces $V_k=[(k-1)c,kc]$ for $k<\ell$
and $V_\ell=[(\ell-1)c, 1]$, with the checkerboard being 1 outside the main
diagonal, 0 on the main diagonal except the upper right corner, and
corresponding to a zero-triangle graph in the upper right corner.

Limits of such graphons in the metric must take the form
\begin{equation} g(x,y) = \begin{cases} 1 & x< kc < y \hbox{ or } y<kc < x \hbox{ for an integer }k<\ell; \cr 
0 & (k-1) c < x,y < kc \hbox{ for some integer }k<\ell; \cr
\hbox{unspecified} & x,y > (\ell-1)c, \end{cases}\end{equation} 
with
\begin{equation} \iiint \limits_{[(\ell-1)c,1]^3} g(x,y) g(y,z) g(z,x) \, dx\, dy \, dz = 0,\end{equation} 
and with $\iint\limits_{[0,1]^2} g(x,y) \, dx \, dy = e$.  Minimizing
$I(g)$ on such graphons is easy, since all but the upper right corner
of the graphon is fixed.  Applying Theorem \ref{flat} to that corner,
we get

\begin{thm}\label{scallop} If $e>1/2$ and $t$ is the smallest value possible, then the minimum of $I(g)$ on $F_0$ is achieved by
  the graphon
\begin{equation} g_0(x,y) = \begin{cases} 1 & x< kc < y \hbox{ or } y<kc < x \hbox{ for an integer }k<\ell; \cr 
  p & (\ell-1)c < x < [1+(\ell-1)c]/2 < y \hbox{ or } (\ell-1)c < y < [1+
  (\ell-1)c]/2 < x; \cr 0 & \hbox{otherwise,} \end{cases}\end{equation} 
where 
\begin{equation} p= \frac{4c(1-\ell c)}{(1-(\ell-1)c)^2}\end{equation} 
is a number chosen to make $\int\!\!\int_{[0,1]^2} g(x,y) \; dx\;dy =
e$. Furthermore, any other minimizer is equivalent to $g_0$. The minimum value
of $I(g)$ is 
\begin{equation} I(g_0) = \frac{(1-(\ell-1)c)^2}{2}I_0(p). \end{equation}
\end{thm}

\section{Minimizing near the first scallop}

Now that we know the minimizer {\em at} the (lower) boundary, we perturb it to
get a minimizer {\em near} the boundary.

\begin{thm}\label{near} Pick $e<1/2$ and $\epsilon$ sufficiently
  small. Then the graphon
\begin{equation} g(x,y) = \begin{cases} 2e-\epsilon & x<\frac{1}{2} < y \hbox{ or } y<\frac{1}{2} < x \cr \epsilon &\hbox{otherwise,} \end{cases}
\end{equation} 
minimizes the rate function to second order in perturbation theory among 
graphons with $e(g)=e$
and $t(g)=e^3-(e-\epsilon)^3$. For pointwise small variations $\delta g$
of $g$, the second variation in $I(g)$ is bounded from below by 
$\frac{1}{2}\iint \limits_{[0,1]^2}(\delta g(x,y))^2 \, dx \, dy$. 
\end{thm}

\begin{proof} We first consider the first variation in $I(g)$ for general
graphons and derive
the Euler-Lagrange equations. It is easy to check that our specific 
$g$ satisfies these equations. We then consider the second variation in
$I(g)$. Note that the function $I_0$ satisfies
\begin{equation} I_0'(u) = \frac{1}{2} [\ln(u) - \ln(1-u)], \qquad I_0''(u) = \frac{1}{2}
\left [ \frac{1}{u} + \frac{1}{1-u}\right ]\ge 2. \end{equation} 

To find the Euler-Lagrange equations with the constraints that $(e(g), t(g))$ 
are equal to fixed values $(e_0,t_0)$, 
we use Lagrange multipliers and vary the function
$I(g) + \lambda_1 (e(g)-e_0) + \lambda_2 (t(g)-t_0)$. To first order, the
variation with respect to $g$ is 
\begin{eqnarray} \delta I(g) &=& \int_0^1\!\!\! \int_0^1 I_0'(g(x,y)) \delta g(x,y) \, dx \, dy
+ \lambda_1 \int_0^1\!\!\!\int_0^1 \delta g(x,y) \, dx \, dy \\
&& + 3 \lambda_2\int_0^1\!\!\!\int_0^1
h(x,y) \delta g(x,y) \, dx \, dy,
\end{eqnarray}
where we have introduced the auxiliary function
\begin{equation} h(x,y) = \int_0^1 \! g(x,z) g(y,z) \, dz. \end{equation} 
Setting $ \delta I(g)$ equal to zero, we get
\begin{equation} \label{ELeq} I_0'(g(x,y)) = - \lambda_1 - 3\lambda_2 h(x,y). \end{equation} 
Our particular $g(x,y)$ satisfies this equation with 
\begin{equation}\label{lambda2} 3 \lambda_2 = \frac{I_0'(2e-\epsilon)-I_0'(\epsilon)}{2(e-\epsilon)^2}. \end{equation} 

Next we expand $\delta t$ and $\delta I$ to second order in $\delta g$, ignoring $O((\delta g)^3)$ terms.
Since 
\begin{equation} 
\delta t = 3\iint h(x,y) \delta g(x,y) dx \, dy + 3 \iiint
g(x,y) \delta g(x,z) \delta g(y,z) dx \, dy \, dz + O((\delta g)^3), 
\end{equation}
and since
we are holding $e(g)$ and $t(g)$ fixed,
\begin{eqnarray} \delta I =&& 
\iint I_0'(g(x,y)) \delta g(x,y)  dx \, dy \cr &&+ \frac{1}{2}\iint
I_0''(g(x,y)) (\delta g(x,y))^2  dx\, dy \cr
=&& \iint (-\lambda_1 -3\lambda_2 h(x,y)) \delta g(x,y)  dx \, dy \cr && + 
\frac{1}{2}\iint
I_0''(g(x,y)) (\delta g(x,y))^2 dx \, dy \cr
=&& -\lambda_1 \delta e - \lambda_2 \delta t + 3 \lambda_2 \iiint
g(x,y) \delta g(x,z) \delta g(y,z) \,dx\, dy\, dz  \cr && + \frac{1}{2}\iint
I_0''(g(x,y)) (\delta g(x,y))^2 \, dx\, dy \cr 
=&& 3 \lambda_2
\iiint g(x,y) \delta g(x,z) \delta g(y,z) dx \,dy \,dz \cr &&+ \frac{1}{4}
\iint I_0''(g(x,y)) \delta g(x,y)^2 dx \,dy + \frac{1}{4}
\iint I_0''(g(x,y)) \delta g(x,y)^2 dx \,dy.
\end{eqnarray}
We have split the $\iint I'' \delta g^2$ term into two
pieces, as we will be applying different estimates to each piece. 

Since $h(x,y)$ and $I''(g)$ are
piecewise constant, all of our integrals 
break down into integrals over different quadrants. Let $R_1$ and $R_2$ be 
the following subsets of $[0,1]^2$:
\begin{equation} R_1 = \{x,y < 1/2\} \cup \{x,y > 1/2\}, \qquad 
R_2 = \{x<1/2<y\} \cup \{y < 1/2 < x\}.\end{equation} 
For each $z$, we define the functions $f_1(z) = \int_0^{1/2} \delta g(x,z) dx$
and $f_2(z) = \int_{1/2}^1 \delta g(x,z) dx$.  The second variation in $I$ is then
\begin{eqnarray} &&  \frac{1}{4} \iint \limits_{[0,1]^2} I_0''(g) \delta g(x,y)^2 dx\,dy   
+ \frac{I_0''(\epsilon)}{4}\iint \limits_{R_1} \delta g(x,y)^2
dx \, dy + \frac{I_0''(2e-\epsilon)}{4} \iint \limits_{R_2} \delta g(x,y)^2 dx \, dy \cr
&+&  3 \lambda_2  \int_0^1 dz \left [
\epsilon \iint \limits_{R_1}  \delta g(x,z) \delta g(y,z) dx \, dy
 +  (2e-\epsilon) \iint \limits_{R_2} \delta g(x,z) 
\delta g(y,z) dx \, dy \right ] \cr
&=& \frac{1}{4} \iint \limits_{[0,1]^2} I_0''(g(x,y)) \delta g(x,y)^2 dx \, dy 
+ \frac{I_0''(\epsilon)}{4}\iint \limits_{R_1} \delta g(x,z)^2
dx \, dz  + \frac{I_0''(2e-\epsilon)}{4} \iint\limits_{R_2} \delta g(x,z)^2 dx \, dz  \cr 
&& \qquad \qquad +  3 \lambda_2 \int_0^1 \epsilon \left [f_1(z)^2
+ f_2(z)^2)\right ] + 2(2e-\epsilon ) f_1(z) f_2(z)\,   dz 
\end{eqnarray}

Note that by Cauchy-Schwarz,
\begin{eqnarray}
\int_0^{1/2} (\delta g(x,z))^2 dx & \ge & 2 \left ( \int_0^{1/2}
\delta g(x,z) dx\right)^2 = 2 f_1(z)^2 \\
\int_{1/2}^1 (\delta g(x,z))^2 dx & \ge & 2 \left ( \int_{1/2}^1
\delta g(x,z) dx\right)^2 = 2 f_2(z)^2. 
\end{eqnarray}
Since $I_0''(\epsilon)$ and $I_0''(2e-\epsilon)$ are positive,
$\delta I$ is  bounded from below by 
\begin{eqnarray} &&  \frac{1}{4} \iint \limits_{[0,1]^2} I_0''(g(x,y)) \delta g(x,y)^2 dx \, dy 
+\frac{ I''(\epsilon)}{2}\left [\int_0^{1/2} f_1(z)^2 dz
+ \int_{1/2}^1 f_2(z)^2 dz \right ] \cr 
&+& \frac{I_0''(2e-\epsilon)}{2} \left [ \int_0^{1/2}f_2(z)^2 dz + 
\int_{1/2}^1 f_1(z)^2 dz \right ]  \cr 
&+& 3\lambda_2 \int_0^1 dz \left [\epsilon(f_1(z)^2 + f_2(z)^2) 
+ 2(2e-\epsilon)f_1(z)f_2(z) \right ] 
\end{eqnarray}

Collecting terms and applying equation (\ref{lambda2}), this bound becomes 
\begin{eqnarray} \frac{1}{4} \iint \limits_{[0,1]^2} I_0''(g(x,y)) \delta g(x,y)^2 dx \, dy &+&  \int_0^{1/2} \!\!dz [c_1 f_1(z)^2 + c_2 f_2(z)^2 + 2 c_3 f_1(z)f_2(z)] \cr 
&+& \int_{1/2}^1 \!\!dz [c_1 f_2(z)^2 + c_2 f_1(z)^2 + 2 c_3 f_1(z)f_2(z)], \end{eqnarray} 
where
\begin{eqnarray}
c_1 &=& \frac{I_0''(\epsilon)}{2} + \frac{\epsilon (I_0'(2e-\epsilon)-I_0'(\epsilon))}
{2(e-\epsilon)^2} \\
c_2 &=& \frac{I_0''(2e-\epsilon)}{2} + \frac{\epsilon (I_0'(2e-\epsilon)-I_0'(\epsilon))}
{2(e-\epsilon)^2} \\
c_3 &=& \frac{(2e-\epsilon)(I_0'(2e-\epsilon)-I_0'(\epsilon))}{2(e-\epsilon)^2}.
\end{eqnarray}
Note that all coefficients are positive, and that $c_2>1$. As $\epsilon \to 0$, $c_1$ goes to $+\infty$ as $1/\epsilon$, 
while $c_3$ only diverges as $-\ln(\epsilon)$. Since $c_1 c_2 > c_3^2$ for small
$\epsilon$, the integrand for each $z$ is positive semi-definite, so the 
integral over $z$ is non-negative, and we obtain
\begin{equation}\delta I \ge \frac{1}{4} \iint I_0''(g) \delta g^2 \ge \frac{1}{2} 
\iint \delta g(x,y)^2,\end{equation}
where we used the fact that $I_0''(u) \ge 2$ for all $u$. 
\end{proof}

Any global minimizer must
be $O(\epsilon)$ close to $g_0$, and hence $O(\epsilon)$ close to our
specified perturbative minimizer. This means that the only way for them
to differ is through a complicated bifurcation of minimizers
at $g_0$, despite the uniform bounds on $\delta I$ as we
approach the boundary. The difference between these hypothetical minimizers and $g_0$
would not be pointwise small, but would merely be small in an $L^1$ sense. 

For instance,
consider graphons of the form
\begin{equation}
g(x,y) = \begin{cases} p & x<c<y \hbox{ or } y<c<x; \cr \alpha & x,y < c; \cr \beta & x,y > c, \end{cases} 
\end{equation}
where $c$ is a parameter that we will vary and $p$, $\alpha$ and $\beta$ are constants that depend on $c$. For each $c$ sufficiently close to $1/2$, it is possible to find a graphon of this form such that $\iiint g(x,y) g(y,z) g(x,z) dx \, dy \, dz = t$ and $\iint g(x,y) dx\, dy =e$, and such that the 
Euler-Lagrange equations (\ref{ELeq}) are satisifed.  Call this graphon $g_c(x,y)$. A lengthly calculation shows that 
\begin{equation} \left . \frac{\partial^2 I(g_c)}{\partial c^2} \right |_{c=1/2} \ge 16e^2 \end{equation}
for small $t$, indicating that (nearly) bipartite graphs with pieces of unequal size have a higher rate function than $g$. 
This provides strong evidence 
that our perturbative solution is in fact a global minimizer for sufficiently small $t$.



\begin{cor}\label{cor1} Assuming our perturbative solution is the global optimizer,
there is a phase transition near the boundary point $(1/2,0)$ 
between the first and second scallop. \end{cor}

\begin{proof} 
Our perturbative solution yields a formula for the
entropy: 

\begin{equation}\label{entropy} s_{e,t} = -\frac{1}{2} [I_0(\epsilon)
  + I_0(2e-\epsilon)].
\end{equation}

This formula for the entropy cannot be extended
  analytically beyond $e=(1+\epsilon)/2$, as $\partial^2 s/\partial
  e^2$ diverges as $e \to (1+\epsilon)/2$. However, $e=(1+
  \epsilon)/2$ corresponds to $t= (\epsilon^3 + 3\epsilon)/4$, or, using
  the more basic variable $e$,

\begin{equation} t= [(2e-1)^3+3(2e-1)]/4,
\end{equation}
which is in
  the interior of $(e,t)$ space. (Since the graphon $g(x,y)$ is nowhere
  zero, it differs in form from the graphons describing graphs
  with minimal $t$.)  Thus $s_{e,t}$ must fail to be analytic in some
  neighborhood of the first scallop.
\end{proof} 

Of course there must also be a phase transition, presumably from this
bipartite phase to a homogeneous phase, if one fixes 
$e$ and raises $t$, which we see as follows. 

\begin{cor}\label{cor2} Assuming our perturbative solution is the
  global optimizer, there is a phase transition as one raises $t$, for
  any fixed $0< e < 1/2$.\end{cor}

\begin{proof}

Recall from Corollary \ref{cor1} the connection
between $t$ and $\epsilon$:

\begin{equation} 
t= \epsilon^3 - (e-\epsilon)^3.
\end{equation} 
Note that $t$ is an increasing function of $\epsilon$ and reaches the
value $2e^3$ when $\epsilon = 2e$. From equation (\ref{entropy}) for
the entropy we see that it cannot be extended analytically to $t >
2e^3$, yet for $e< 1/2$ we have $2e^3< e^2< e^{3/2}$ so $(e,t)$ is in
the interior of the phase space (Figure 1). 
\end{proof}

\section{Conclusion}

Our goal was to analyze possible phase transitions between
multipartite phases of complex networks, analogous to solid/solid
transitions in materials. To this end we adapted the Strauss model
\cite{St}, defined in the grand canonical ensemble, to a
microcanonical ensemble. It is appropriate at this point to review the
`equivalence of ensembles' in statistical physics.

In thermodynamics the concavity of the entropy $S(E,N,V)$, as a
function of internal energy $E$, particle number (or mass) $N$ and
volume $V$, and the interpretation of equilibrium states as states
maximizing the entropy, are both fundamental; see for instance
\cite{Ca, Ma}. Lagrange multipliers can be used to convert this
optimization criterion of the entropy to an equivalent optimization
criterion of the free energy, where the free energy is the Legendre
transform of the entropy \cite{Do, Ca, Ma}. It is important that the
Legendre transform between the entropy and free energy be invertible
so the two optimization schemes are equivalent, and this follows from
the concavity of the entropy. 
(See Section 26 of \cite{Ro} for the mathematics of the
Legendre tranform between convex functions.)

Statistical mechanics supplies a model for thermodynamic states, as
probability distributions on mechanical multiparticle states. From a
given short range particle interaction one can then (in principle)
compute the internal energy $E$ and entropy $S$, and prove the above
two features of the entropy: its concavity and its optimization role
for equilibrium states. To do this one uses the basic Boltzmann/Gibbs
ansatz: that the entropy $S(E,N,V)$ is proportional to
\begin{equation}
-\sum_j\rho_j\ln(\rho_j)
\end{equation} 
where $\rho_j$ is the probability of multiparticle state $j$, and the
equilibrium state is that probability distribution $\{\rho_j\}$ on the
set $\Xi(E,N,V)$, of multiparticle states of energy $E$ and particle number $N$
in volume $V$, which maximizes the entropy \cite{R1, R2,
  Ge, Wi, Ma}. 
(Note that in taking an infinite volume limit, which
we must do to obtain equivalence of ensembles, one can divide the
entropy's variables by volume, and consider the entropy density as a
function of particle and energy density).

The equivalence of ensembles in statistical mechanics is basically a
strenghening of the equivalence in thermodynamics between entropy
$s_{e,t}$ and free energy $\psi_{\beta_1,\beta_2}$, corresponding to
Lagrange multipliers $\beta_1$ and $\beta_2$, which follows from the
concavity of the entropy. With the modeling of the thermodynamic
states this now implies a bijection $(e,t)\longleftrightarrow
(\beta_1,\beta_2)$ such that $s_{e,t}$ and $\psi_{\beta_1,\beta_2}$
have the same optimizing states, at least off some manageable sets of
parameter pairs corresponding to `phase coexistence' where the
bijection can degenerate to a many-to-one map \cite{Ge}.

In exponential random graph models, which are mean field rather than
short range, the entropy need not be concave \cite{TET} and indeed
this fails in an obvious way for the specific model we are analyzing,
the Strauss model, since even the domain $R$ of the entropy is not
convex (see Figure 1). Therefore in the infinite node limit of the
model the free energy density $\psi_{\beta_1,\beta_2}$,  need not be
equivalent to the entropy density $s_{e,t}$; $\psi_{\beta_1,\beta_2}$
can be obtained from $s_{e,t}$ by Legendre transform, but it may not
be possible to recover $s_{e,t}$ from $\psi_{\beta_1,\beta_2}$.
Inequivalence can result from the existence of graphons maximizing
$s_{e,t}$ for some $(e,t)$ which are not maximizers of
$\psi_{\beta_1,\beta_2}$ for any $(\beta_1,\beta_2)$.  Specific
instances of such loss of information in $\psi_{\beta_1,\beta_2}$ are
shown in a future paper \cite{RS}, but one consequence can already be
seen in the transitions, studied previously \cite{PN, CD, RY} in the
grand canonical ensemble, between independent-edge graphs across a
phase transition curve in the phase space; see Figure 2 for the
Strauss model. Such `free particle' graphs, with only edge density $e$
as a variable, optimize $\psi_{\beta_1,\beta_2}$ for
$(\beta_1,\beta_2)$ throughout the upper half of the grand canonical
phase space, so $e$ is a function of $(\beta_1,\beta_2)$ off the
transition curve there; see \cite{RY} for details. These graphs all
lie on the curve $t=e^3$ in Figure 1, not a 2-dimensional region in
that microcanonical phase space, making it difficult to use
singularities of the free energy $\psi_{\beta_1,\beta_2}$ to imply
singularities of the entropy $s_{e,t}$.  For this reason we have
focused here on phase transitions in the lower region of the
microcanonical phase space, Figure 1.

\begin{figure}[h]
\vskip.4truein
\includegraphics[width=3in]{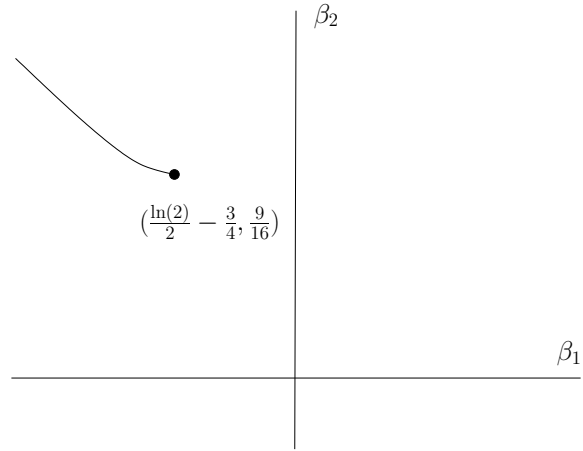}
\caption{The curve of all singularities of $\psi_{\beta_1,\beta_2}$
  for $\beta_2>-1/2$}
\label{trans}
\end{figure}

We have shown that our (\hbox{2-parameter}, bipartite) graphons $g$ of
Theorem \ref{near} maximize the entropy density at least to second
order in perturbation theory, among graphons with a limited range of
edge and triangle densities. Assuming the $g$ are actually global
maximizers we then proved the entropy density would have to lose
analyticity as the edge density of the graphon approaches the
tripartite regime. We also show that the entropy density must suffer a
phase transition as the triangle density is raised sufficiently high,
presumably from the structured bipartite phase to a homogeneous phase
of higher triangle density.

We expect that a more complicated analysis could produce appropriate
graphons $g^{(k)},\ k\ge 1$, near each of the higher edge density
(multipartite) graphons of minimial triangle density, with a transition
near each scallop intersection.  Intuitively this suggests a mechanism
whereby as edge density is increased, near the minimum triangle density
graphon, the system progressively transitions through finer and finer
structure; for high edge density most graphs would consist of
many interacting `parts'.

Our results on phase transitions require that the graphons of Theorem
\ref{near} be in fact global, not just local, maximizers of the
entropy density.  In a future paper \cite{RS} we use a symmetry to prove that
these graphons are indeed the unique global maximizers at least for
triangle density in the range $0\le t\le 1/8$ and edge density
$e=1/2$, and we can then see a transition on this curve. However we
still cannot prove the graphons are the global optimizers of entropy
density for $(e,t)$
in any two-dimensional region, as is needed to fully justify the notion
of a structured phase. (See \cite{AR} for a variant of this approach.)

In conclusion we emphasize that our key tool was Theorem \ref{thm1},
an optimization formula for the asymptotic entropy density, and made
essential use of
the graph limit formalism.  The graphon formalism is a powerful tool
for dealing with the infinite size limit in mean field models, and
we have used it to make some progress on understanding the
structure of asymptotically large graphs near the extreme of low
triangle density.

\medskip

\noindent {\bf Acknowledgements:} We gratefully acknowledge useful 
discussions with Francesco Maggi and Peter Winkler.


\begin{thebibliography}{BCLSV}


%

\bibitem[AR]{AR} D. Aristoff and C. Radin, Emergent structures in large 
networks, J. Appl. Probab. (to appear), arXiv:1110.1912

\bibitem[BCL]{BCL} {C. Borgs, J. Chayes and L. Lov\'{a}sz}, Moments of
  two-variable functions and the uniqueness of graph limits,
  Geom. Funct. Anal. 19 (2010) 1597-1619.

\bibitem[BCLSV]{BCLSV} {C. Borgs, J. Chayes, L. Lov\'{a}sz, V.T. S\'os
    and K. Vesztergombi}, Convergent graph sequences I: subgraph
  frequencies, metric properties, and testing, { Adv. Math.} { 219}
  (2008) 1801-1851.



\bibitem[Ca]{Ca} H.B. Callen, {Thermodynamics}, John Wiley, New York, 1960.

\bibitem[CD]{CD} S. Chatterjee, and P. Diaconis, Estimating and understanding
exponential random graph models, arXiv: 1102.2650v3.

\bibitem[CV]{CV} S. Chatterjee and S.R.S. Varadhan, The large deviation principle
for the Erd\H{o}s-R\'{e}nyi random graph, Eur. J. Comb. 32 (2011)
1000-1017

\bibitem[Do]{Do} T.C. Dorlas, Statistical Mechanics: Fundamentals and
  Model Solutions, Institute of Physics Publishing, London, 1999.


\bibitem[Ge]{Ge} H-O. Georgii, The equivalence of ensembles for
  classical systems of particles, J. Stat. Phys. 80 (1995) 1341-1378.



\bibitem[Lov]{Lov} L. Lov\'asz, Large networks and graph limits,
American Mathematical Society, Providence, 2012. 

\bibitem[LS1]{LS1} {L. Lov\'{a}sz  and  B. Szegedy},
Limits of dense graph sequences, 
{ J. Combin. Theory Ser. B} { 96} (2006)  933-957.

\bibitem[LS2]{LS2} {L. Lov\'{a}sz  and  B. Szegedy},
Szemer\'edi's lemma for the analyst,
{ GAFA} { 17} (2007)  252-270.

\bibitem[LS3]{LS3} {L. Lov\'{a}sz  and  B. Szegedy},
Finitely forcible graphons,
{ J. Combin. Theory Ser. B} {101} (2011) 269-301.


\bibitem[Ma]{Ma} S.-K. Ma, {Statistical Mechanics}, World Scientific,
Singapore, 1985.

\bibitem[Ne]{Ne} M.E.J. Newman, Networks: an Introduction, Oxford
  University Press, 2010.



\bibitem[PN]{PN} J. Park and M.E.J. Newman, Solution for the properties of a
clustered network, Phys. Rev. E 72 (2005) 026136.


\bibitem[PR]{PR} O. Pikhurko and A. Razborov, Asymptotic structure of graphs with
the minimum number of triangles, arXiv:1203.4393

\bibitem[R1]{R1} D. Ruelle, {\it Statistical Mechanics; Rigorous
  Results}, Benjamin, New York, 1969.

\bibitem[R2]{R2} D. Ruelle, {\it Thermodynamic Formalism}, Addison-Wesley, New
York, 1978. 

\bibitem[Ro]{Ro} R.T. Rockafellar, Convex Analysis, Princeton
  University Press, Princeton, 1970.

\bibitem[RS]{RS} C. Radin and L. Sadun, Singularities in the entropy
of asymptotically large simple graphs, arXiv:1302:3531

\bibitem[RY]{RY} C. Radin and M. Yin, Phase transitions in exponential random
graphs, Ann. Appl. Probab. (to appear), arXiv:1108.0649.


\bibitem[St]{St} D. Strauss, On a general class of models for interaction, SIAM
Rev. 28 (1986) 513-527.


\bibitem[TET]{TET} H. Touchette, R.S. Ellis and B. Turkington, Physica
  A 340 (2004) 138-146.


\bibitem[Wi]{Wi} A.S. Wightman, Convexity and the notion of
  equilibrium state in thermodynamics and statistical mechanics,
  in R. Israel, {Convexity in the Theory of Lattice Gases}, 
Princeton University Press, Princeton, 1979, pp. ix-lxxv.


\end{thebibliography}
\end{document}